%% file: main.tex
\documentclass{article}
\usepackage{ijcai24}

\usepackage{times}
\usepackage{soul}
\usepackage{url}
\usepackage{amsthm,amsmath,amsfonts}
\usepackage[hidelinks]{hyperref}
\usepackage[utf8]{inputenc}
\usepackage[small]{caption}
\usepackage{graphicx}
\usepackage{amsmath}
\usepackage{amsthm}
\usepackage{amssymb}
\usepackage{booktabs}
\usepackage{algorithm}
\usepackage{algorithmic}
\usepackage[switch]{lineno}
\usepackage{xspace,xcolor}
\usepackage{bm}
\usepackage{mathtools}
\usepackage{delimset}
\usepackage{tikz}

\allowdisplaybreaks
\allowdisplaybreaks[3]

\newcommand{\acc}{\ensuremath{{\tt{acc}}}\xspace}

\newcommand{\E}{\bm{\mathbb{E}}}
\newcommand{\Z}{\mathbb{Z}}

\DeclareMathOperator*{\argmin}{arg\,min}

\newtheorem{theorem}{Theorem}
\newtheorem{lemma}[theorem]{Lemma}
\newtheorem{corollary}[theorem]{Corollary}
\newtheorem{proposition}[theorem]{Proposition}
\newtheorem{definition}[theorem]{Definition}

\title{Contract Scheduling with Distributional and Multiple Advice}

\author{
Spyros Angelopoulos$^1$\and
Marcin Bienkowski$^2$\and
Christoph Dürr$^1$\and
Bertrand Simon$^3$\\
\affiliations
$^1$Sorbonne Université, CNRS, LIP6\\
$^2$University of Wroclaw\\
$^3$CNRS, IN2P3 Computing Center\\
\emails
\{spyros.angelopoulos, christoph.durr\}@lip6.fr,
marcin.bienkowski@cs.uni.wroc.pl,
bertrand.simon@cnrs.fr
}

\begin{document}

\maketitle

\begin{abstract}
Contract scheduling is a widely studied framework for designing real-time systems with interruptible capabilities. Previous work has showed that a prediction on the interruption time can help improve the performance of contract-based systems, however it has relied on a single prediction that is provided by a deterministic oracle. In this work, we introduce and study more general and realistic learning-augmented settings in which the prediction is in the form of a probability distribution, or it is given as a set of multiple possible interruption times. For both prediction settings, we design and analyze schedules which perform optimally if the prediction is accurate, while simultaneously guaranteeing the best worst-case performance if the prediction is adversarial. We also provide evidence that the resulting system is robust to prediction errors in the distributional setting. Last, we present an experimental evaluation that confirms the theoretical findings, and illustrates the performance improvements that can be attained in practice.
\end{abstract}

\section{Introduction}
\label{sec:introduction}

A central requirement in the design of real-time and intelligent systems is the provision for anytime capabilities. Many applications, such as motion-planning and medical diagnosis, require systems that are able to output a reasonably efficient solution even if they are interrupted at arbitrary points in time. This motivates the design and evaluation of anytime systems given, as building component, more rudimentary systems that are not interruptible. Questions related to trade-offs between  resources (e.g., computational time) and performance are at the heart of AI applications, and the topic of early influential works on flexible computation, resource-bounded algorithms and time-depending planning~\cite{deliberation:boddy.dean,Horvitz:reasoning,DBLP:journals/ai/ZilbersteinR96}.

We are interested, in particular, in a paradigm introduced by~\cite{RZ.1991.composing}, in which the building component is a {\em contract} algorithm, namely an anytime algorithm that is given the exact amount of allowable computation time as part of its input. A contract algorithm will always output the correct result if it is allowed at least its prescribed computation time (hence its name), otherwise it may very well return a meaningless result if it is queried prior to its promised contract time. Thus, contract algorithms are not interruptible; however~\cite{RZ.1991.composing} proposed a methodology for obtaining interruptible systems via {\em scheduling} consecutive executions of the contract algorithm with increasing computation times. As an example, consider a schedule in which the $i$-th execution of the contract algorithm is allowed time $2^i$, for all $i \in \mathbb{N}$. This yields a system in which, at any time $t$, the contract algorithm has completed at least one execution of time $t/4$. The factor 4 quantifies the performance of the schedule, and is the multiplicative loss due to the repeated executions of the contract component.

More generally, given a contract algorithm $A$, a contract {\em schedule} is defined as an increasing sequence $X=(x_i)_{i \geq 0}$, where $x_i$ is the length of the $i$-th execution of $A$. To evaluate the performance of a schedule $X$, we rely on a worst-case measure known as the {\em acceleration ratio}~\cite{RZ.1991.composing}. Let $\ell(X,T)$ denote the largest contract length completed by time $T$ in $X$, then the acceleration ratio of $X$ is defined formally as
\begin{equation}
\acc(X)=\sup_T \frac{T}{\ell(X,T)}.
\label{eq:acc.ratio}
\end{equation}

It is easy to show that the optimal acceleration ratio is equal to 4, and that it is obtained by the doubling schedule $x_i=2^i$. However, contract scheduling becomes far more challenging in more complex settings that have been studied in the literature. This includes schedules that are run on multiple processors~\cite{BPZF.2002.scheduling}, on multiple instances~\cite{ZilbersteinCC03}, or combinations of these seeings~\cite{steins,aaai06:contracts}; soft interruptions where the query time is not a hard deadline~\cite{soft-contracts}; performance measures beyond the acceleration ratio~\cite{ALO:multiproblem}; adaptive schedules~\cite{DBLP:journals/jcss/AngelopoulosP23}; schedules with end-guarantees on completion time~\cite{DBLP:conf/ijcai/0001J19}; and learning-augmented schedules~\cite{DBLP:journals/jair/AngelopoulosK23}. Furthermore, contract scheduling is an abstraction of resource allocation under uncertainty, hence it has connections to other optimization problems under uncertainty, such as searching for a hidden target in a known environment under the competitive ratio, as shown in~\cite{steins,spyros:ijcai15}. 

\subsection{Contract scheduling with predictions}

The standard formulation of contract scheduling assumes that the scheduler has no prior information on the interruption time, hence the acceleration ratio evaluates the system performance at worst-case interruptions. In practice, however, one expects that the scheduler should be able to benefit from a {\em prediction} on the interruption, that may be available via a ML oracle. This motivated the recent study~\cite{DBLP:journals/jair/AngelopoulosK23} of the problem, in a setting in which an imperfect oracle provides the system with a single, deterministic prediction (advice) $\tau$ concerning the interruption time $T$. \cite{DBLP:journals/jair/AngelopoulosK23} showed that there exists a schedule with acceleration ratio equal to 4 even if the prediction is adversarial, and which also yields a much improved ratio equal to 2 if the prediction is correct (i.e., error-free). Moreover, they showed that this is result is optimal. Using the terminology of {\em learning-augmented} online algorithms~\cite{DBLP:journals/jacm/LykourisV21,NIPS2018_8174}, we say that there exist schedules of {\em consistency} equal to 2, and {\em robustness} equal to 4, whereas no 4-robust schedule can be better than 2-consistent\footnote{\cite{DBLP:journals/jair/AngelopoulosK23} also studied a query-based setting in which the prediction is elicited via responses to binary queries, however, this query model is not relevant to our work.}. We emphasize that the prediction is assumed to be {\em single}, i.e., it consists of a unique predicted interruption time, and deterministic, in that the oracle does not incorporate any stochastic aspects. 

In practice, however, the above model may not adequately capture the information content of the anticipated prediction. For instance, in motion-planning algorithms, the system may operate under a certain probabilistic knowledge of the terrain, hence specific actions may have to be triggered according to a (stochastic) belief about the environment~\cite{DBLP:conf/ijcai/ZilbersteinR93}. For a different example, a medical diagnostic system may have to be queried at multiple anticipated times (e.g., depending on the availability of various facilities and specialists). Note that such settings are not captured by the model of~\cite{DBLP:journals/jair/AngelopoulosK23}, in which the schedule is fine-tuned according to a single prediction, and may have consistency as high as 4 if the prediction includes as few as two possible interruption times.

\subsection{Contribution}
\label{subsec:contribution}

Motivated by the above applications, and the limitations of single/deterministic prediction oracles, we study contract scheduling under more general models that incorporate {\em distributional} and {\em multiple} advice. As in~\cite{DBLP:conf/infocom/ImMXZ23}, which studied the learning-augmented dynamic acknowledgment problem, we aim to {\em simultaneously} optimize the consistency and the robustness of the system. For both settings, we show that the results we obtain are tight.

We begin with the distributional setting, in which the advice oracle provides the scheduler with a distribution on the anticipated prediction. Here, the consistency is evaluated, in expectation, relative to the distributional advice (see Section~\ref{sec:preliminaries} for the formal definition). We show how to construct a collection of $n$ schedules for any given $n \in \mathbb N^*$, such that each schedule in the collection is 4-robust, and the best schedule has consistency at most $4n(2^\frac{1}{n}-1)$. In particular, we show that as $n$ grows, the consistency of the system is arbitrarily close to $4\ln 2 \approx 2.77$, and that the best schedule can be computed in time polynomial in~$n$. We show that this bound is optimal, in that there exists a distributional prediction for which the consistency of any 4-robust schedule is at least $4 \ln 2$. Furthermore, we demonstrate an interesting disconnect between deterministic and distributional predictions. Namely, we prove that, given the distributional prediction that maximizes the consistency, the performance of the optimal schedule deteriorates {\em smoothly} as function of the prediction error, measured by the Earth Mover's Distance, or EMD~\cite{rubner1998metric}. In contrast, no consistency-optimal 4-robust schedule can exhibit smoothness against deterministic predictions.  This disconnect shows that distributional predictions can help mitigate pathological situations, which can be of interest in other learning-augmented optimization problems. 

In the second part of this work, we study the model in which the advice oracle provides the scheduler a set $P$ of $k$ potential interruption times (e.g., provided by $k$ experts). Here, the consistency is measured as the worst-case performance ratio among interruptions in $P$, and we refer to Section~\ref{sec:preliminaries} for the formal definition. We show how to derive a 4-robust schedule of optimal consistency $2^{2-\frac{1}{k}}$ in time $O(k \log k)$. 
We conclude with an experimental evaluation of our schedules, in both the distributional and multiple advice settings, that demonstrates the performance improvements that can be attained in practice.

\subsection{Other related work}

Motivated by the capacity of ML predictions to improve algorithmic performance, the field of learning-augmented algorithms has been growing rapidly in the recent years. We refer to the survey~\cite{DBLP:books/cu/20/MitzenmacherV20} and the online repository~\cite{predictionslist} that lists several works over the last five years. The vast majority of works have focused on single, deterministic predictions. Multi-prediction oracles were first studied in the context of ski rental~\cite{gollapudi2019online}, followed by works on multi-shop ski rental~\cite{wang2020online}, facility location~\cite{almanza2021online}, matching and scheduling~\cite{dinitz2022algorithms}, online covering~\cite{anand2022online} and $k$-server~\cite{antoniadis2023mixing}. Distributional predictions were first studied in~\cite{diakonikolas2021learning} in problems such as ski rental and prophet inequalities. 

Contract scheduling is related to the {\em online bidding} problem~\cite{ChrKen06}, which is used a building block in many 
online algorithms, e.g., for scheduling on related machines, both non-preemptive~\cite{DBLP:journals/jacm/AspnesAFPW97},
and preemptive~\cite{DBLP:journals/scheduling/EbenlendrS09}, $k$-median problem~\cite{DBLP:journals/algorithmica/ChrobakKNY08}, 
or multi-level aggregation~\cite{DBLP:journals/tcs/BienkowskiBBCDF21}.
A reduction between the contract scheduling and online bidding problems for the single prediction setting was shown in~\cite{DBLP:journals/jair/AngelopoulosK23}. Consistency/robustness tradeoffs for online bidding with a single prediction were shown in~\cite{anand2021regression,DBLP:conf/innovations/0001DJKR20,DBLP:conf/pkdd/ImMXZ23}.

\section{Preliminaries}
\label{sec:preliminaries}

A contract schedule is defined as a sequence of the form $X=(x_i)_{i\in \mathbb{N}}$, where $x_i$ is the {\em length} of the $i$-th contract, and recall that the acceleration ratio of $X$ is given by~\eqref{eq:acc.ratio}. In a learning-augmented setting, the acceleration ratio of $X$ is equivalently called the {\em robustness} of $X$, and we say that $X$ is $r$-robust if it has robustness at most $r$. From~\cite{RZ.1991.composing} we thus know that $r\geq 4$, and that for $X=(2^i)_{i \in \mathbb{N}}$, we have that $r(X)=4$.

The {\em consistency} of $X$ is defined according to the specifics of the prediction oracle, hence we make a distinction between the two learning-augmented settings we study. In the distributional setting, the advice consists of a distribution $\mu$ on the anticipated interruption time, and the consistency of a schedule $X$ with advice $\mu$ is defined as
\begin{equation}
c(X,\mu) = \frac{\E_{z \sim \mu}[z]}{\E_{z \sim \mu}[\ell(X,z)]},
\label{eq:cons.distr}
\end{equation}
and recall that the random variable $\ell(X,z)$ denotes the largest contract completed in $X$ by time $z$. We will call $\ell(X,z)$ the {\em profit} of $X$, where $z$ is drawn from distribution $\mu$.

In the multiple-advice setting, the prediction consists of a set $P=\{\tau_1, \ldots ,\tau_k\}$ of $k$ possible interruption times (e.g., provided by $k$ experts). Here, we measure the consistency $c(X,P)$ of a schedule $X$ with prediction set $P$ as its worst-case performance relative to the prediction in $P$, i.e., we define
\begin{equation} 
c(X,P)=\sup_{\tau \in P} \frac{\tau}{\ell(X,\tau)}.
\label{eq:cons.mult}
\end{equation}

Note that without any assumptions, no schedule can have bounded robustness, if the interruption time is allowed to be arbitrarily small. There are two types of assumptions that can be applied to circumvent this technical issue. The first is to assume that the interruption can only occur after the first contract has completed its execution. The second is to assume that the schedule is {\em bi-infinite}, in that it starts with an infinite number of infinitesimally small contracts. For instance, the doubling schedule can be described as $(2^i)_{i\in \mathbb{Z}}$, and the completion time of contract $i\geq 0$ is defined as $\sum_{j=-\infty}^i 2^j=2^{i+1}$. We choose the second assumption since it simplifies the calculations, but we note that the two assumptions can be used interchangeably; see, e.g., the discussion in~\cite{demaine:turn}.

For a given $\lambda \in [0,1)$, define the schedule $X(\lambda) \triangleq (2^{i-\lambda})_{i \in \mathbb{Z}}$. The following proposition shows that it suffices to focus on the set of schedules $\cup_{\lambda \in [0,1)} \{X(\lambda))\}$.

\begin{proposition}[Appendix]
For any $\lambda \in [0,1)$, $X(\lambda)$ is 4-robust. Conversely, every 4-robust schedule must belong in 
the class  $\cup_{\lambda \in [0,1)} \{X(\lambda)\}$.
\label{lem:4-robust}
\end{proposition}


\section{Distributional advice}

We begin with the setting in which the advice is in the form of a given distribution $\mu$. We first define an appropriate collection $S_n$ of $n$ schedules which will be instrumental towards finding an optimal schedule. 

\begin{definition}
For any $n$, let $S_n$ denote the following collection of $n$ schedules $X_0
\ldots ,X_{n-1}$, defined as $X_j=X(j/n)=(2^{i-j/n})_{i \in \mathbb{Z}}$.
\label{def:portfolio}
\end{definition}

\begin{theorem}
For any $n \in \mathbb{N}^+$, there exists a 4-robust schedule in $S_n$ that has consistency at
most $4n \cdot (2^{1/n}-1)$.
\label{thm:upper.dist}
\end{theorem}

\begin{proof}
First, observe that any schedule $X_j$ satisfies the conditions 
of Proposition~\ref{lem:4-robust}, and hence is $4$-robust. Thus, it remains 
to show that for any distribution $\mu$ over $(0,+\infty)$, 
there exists a schedule in $S_n$ with the desirable
consistency. 

We define the intervals $I_i^k$ for $i \in \Z$ and 
$k \in \{0, \dots, n-1\}$ as $I_i^k = [2^{i+ k/n}, 2^{i +(k+1)/n})$.
Note that the intervals $I_i^k$ are disjoint and their union is equal to the 
support of $\mu$, i.e.,
$\biguplus_{i \in \Z} \biguplus_{0 \leq k \leq n-1} I_i^k = (0,+\infty)$.

If the interruption time $z$ is drawn from distribution $\mu$, then
its expected value is 
\begin{align*}
\E_{z\sim\mu}[z] 
    & \leq \sum_{i \in \Z}
        \sum_{k=0}^{n-1} 2^{i+(k+1)/n}\cdot \mu(I_i^k) \\
    & \leq \sum_{i \in \Z} 2^i \cdot \sum_{k=0}^{n-1} \mu(I_i^k) \cdot 2^{k/n} 
    \cdot 2^{1/n}.
\end{align*}

\begin{figure}
\centering
\includegraphics[width=0.9\columnwidth]{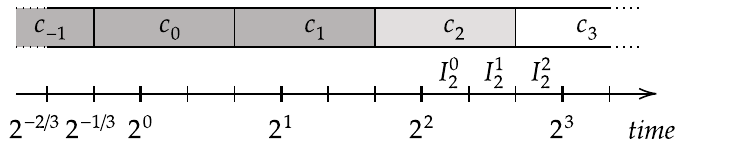}
\caption{Illustration of the computation of $\ell(X_j,z)$ for $n=3$ and $j=1$.
Note that the time scale is logarithmic. Fix an interruption point $z \in [2^2, 2^3)$. 
Then $z$ is contained in $I_2^k$ for some value of $k \in \{0,\dots,n-1\}$.
If $k \geq n - j = 2$, then the largest completed contract is $c_2$ (of length $2^{2-1/3}$),
and otherwise, the largest completed contract is $c_1$ (of length $2^{1-1/3}$).}
\label{fig:upper-dist}
\end{figure}

We now estimate the expected value of $\ell(X_j,z)$ for each schedule~$X_j$ with 
$j \in \{0, \dots, n-1\}$.
To this end, we consider a fixed value of interruption time $z$ and let $I_i^k$ be the interval containing $z$. 
We note that if $k\geq n-j$, then $2^{i+1-j/n}\leq z < 2^{i+2-j/n}$, 
and thus $\ell(X_j,z) = 2^{i-j/n}$.
Otherwise, i.e., if $k<n-j$, we have $\ell(X_j,z) = 2^{i-1-j/n}$.
See \autoref{fig:upper-dist} for a pictorial description of this computation.

Hence,
\begin{align*}
\E_{z\sim\mu}[\ell(X_j,z)] 
    & = \sum_{i \in \Z} \sum_{k=0}^{n-j-1}2^{i-1-j/n}
        \cdot\mu(I_i^k) \\
    & + \sum_{i \in \Z} \sum_{k=n-j}^{n-1}2^{i-j/n}\cdot\mu(I_i^k).
\end{align*}

Next, we compute the sum of the largest contracts over all schedules $X_j$, as follows:
\begin{align*}
& \sum_{j=0}^{n-1} \E_{z\sim\mu}[\ell(X_j,z)] \\
&= \sum_{j=0}^{n-1} \sum_{i \in \Z} 
    \left(\sum_{k=0}^{n-j-1}2^{i-1-\frac jn} \cdot \mu(I_i^k) 
    + \sum_{k=n-j}^{n-1}2^{i-\frac jn}\cdot\mu(I_i^k)\right)\\
&=  \sum_{i \in \Z} 2^i \cdot \sum_{k=0}^{n-1} \mu(I_i^k) \cdot 
    \left(\frac{1}{2} \sum_{j=0}^{n-k-1}2^{-j/n} + \sum_{j=n-k}^{n-1} 2^{-j/n}\right) \\
    \displaybreak[1]
&=  \sum_{i \in \Z} 2^i \cdot \sum_{k=0}^{n-1} \mu(I_i^k)\cdot 
    \left(
        \frac{1}{2} \cdot \frac{1-2^{\frac{k-n}{n}}}{1-2^{-\frac 1n}} 
        + 2^{\frac{k-n}{n}}\cdot\frac{1-2^{\frac{-k}{n}}}{1-2^{-\frac 1n}}
    \right) \\
&=  \sum_{i \in \Z}2^i \cdot \sum_{k=0}^{n-1} \mu(I_i^k) 
    \cdot \frac{2^{k/n-2}}{1-2^{-1/n}} \\
     \displaybreak[1]
&=  \frac{\E_{z \sim \mu}[z]}{2^{1/n}} \cdot \frac{2^{-2}}{1-2^{-1/n}} 
=  \frac{\E_{z \sim \mu}[z]}{4  (2^{1/n} - 1)}.
\end{align*}

We can now use an averaging argument, and deduce that there exists a schedule $X_{j^*}$ in $S_n$ such that 
\begin{align*}
    \E_{z\sim\mu}[\ell(X_{j^*},z)] 
    & \geq \frac{1}{n} \cdot \sum_{j=0}^{n-1} \E_{z\sim\mu}[\ell(X_j,z)] \\
    &\geq  \frac{\E_{z \sim \mu}[z]}{4n  (2^{1/n} - 1)}. 
\end{align*}
Hence, the consistency of the schedule $X_{j^*}$ is 
$c(X_{j^*},\mu) = \E_{z \sim \mu}[z] / \E_{z \sim \mu}[\ell(X_{j^*},z)] 
\leq 4n (2^{1/n} - 1)$.
\end{proof}

We observe that $4n  (2^{1/n}-1)$ is decreasing and tends to $4  \ln 2$ as $n \to \infty$. Indeed, $4n  (2^{1/n}-1) = 4n  (e^{\ln 2 \cdot
(1/n)}-1) = 4n  (\ln 2 \cdot (1/n) + O(1/n^2)) = 4  \ln 2 + O(1/n)$.
This means that for any $\epsilon > 0$, there exists an integer $n = O(1/\epsilon)$ such
that $4n (2^{1/n}-1) \leq 4  \ln 2+\epsilon$. Assuming that we can evaluate the probability associated with an interval  in constant time, we obtain the following corollary.

\begin{corollary}
For any arbitrarily small $\epsilon>0$, 
there is an algorithm with runtime polynomial in $O(1/\epsilon)$
for devising a~4-robust schedule that has consistency at most $4 \cdot \ln 2+\epsilon$.
\label{cor:upper.dist}
\end{corollary}

We now show that the upper bound of Theorem~\ref{thm:upper.dist} is tight.

%
%

\input{distrib-LB}

\subsection{Smoothness with prediction errors}

We will now discuss an interesting disconnect between deterministic and distributional predictions that arises in contract scheduling. Consider first the setting of a single, deterministic prediction, for which we know that there exist 4-robust, 2-consistent schedules~\cite{DBLP:journals/jair/AngelopoulosK23}. Given such prediction, say $\tau$, let $\eta=|T-\tau|$ denote the prediction {\em error} associated with $\tau$, where $T$ denotes the actual interruption.

\begin{proposition}[Appendix]
For any arbitrarily small $\epsilon>0$, and any schedule $X$ with prediction $\tau$ that is 4-robust and 2-consistent, there exists $T$ such that $\eta=|T-\tau|=\epsilon$ and $\ell(X,T) \leq \frac{1}{4}T+\frac{\epsilon}{4}$.
\label{prop:single}
\end{proposition}

Proposition~\ref{prop:single} shows that, in the single prediction setting, any schedule that simultaneously optimizes the consistency and the robustness  is extremely fragile with respect to prediction errors. Informally, in the presence of a tiny prediction error, the proposition shows that the consistency of the schedule becomes as large as the robustness, namely arbitrarily close to 4, hence the single prediction leads to no improvement in any practical situation.

\input{smoothness.tex}

\section{Multiple advice} \label{sec:multiple}

In this section, we study the setting in which the prediction oracle provides a set $P$ that consists of $k$ (potential interruption) times, denoted by $\tau_0, \ldots ,\tau_{k-1}$. Recall that the consistency of a schedule is given by~\eqref{eq:cons.mult}. For every $j \in [0,k-1]$, we denote by $\delta_j \in [0,1)$, and by $i_j \in \mathbb{N}$ the unique values such that $\tau_i=2^{i_j+\delta_j}$. We assume, without loss of generality that $0\leq \delta_0 \leq \delta_1 \ldots \leq \delta_{k-1}<1$.

We first show that we can find, in time polynomial in $k$, a schedule that simultaneously minimizes the consistency and the robustness.

\begin{theorem}
Given a prediction set $P$ of size $k$, we can find a 4-robust schedule of optimal consistency in time $O(k^2)$.
\label{thm:multiple.optimal}
\end{theorem}

\begin{proof}
We first argue that there exists a 4-robust schedule of optimal consistency $X^*=(x_i^*)_{i \in \mathbb{Z}}$, such that $X^*$ contains at least one contract whose execution is completed at time $\tau_j$, for some 
$\tau_j \in P$. By way of contradiction, suppose this is not the case. From Proposition~\ref{lem:4-robust}, the last contract completed by $X^*$ by time $t$ finishes at time $2\ell(X^*,t)$.  Define $ a=\min_{j\in [0,k-1]} \frac{\tau_j}{2\ell(X^*,\tau_j)}$, and consider the schedule  $X'=(2^{\log_2 a}\cdot x_i^*)_{i \in \mathbb{Z}}$.  From Proposition~\ref{lem:4-robust}, 
$X'$ is also 4-robust. Moreover, we have that $\ell(X',\tau_j) \geq \ell(X^*,\tau)$, therefore $X'$ is also optimal, a contradiction.


\input{circle-figure}
\begin{figure}[t!]
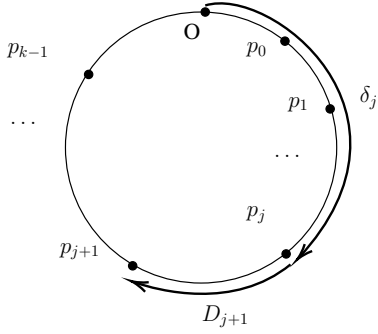

\centering
\scalebox{0.6}{\snap}
\caption{Illustration of the definitions of $\delta_j$ and 
$D_j$.}
\label{fig:snapshot}
\end{figure}

The above observation leads to the following algorithm for finding an optimal schedule. For every $j \in [0,k-1]$, define the schedule $X_j=(2^{i+\delta_j})_{i \in \mathbb{Z}}$, and note that, by definition, $X_j$ has a contract that terminates at time $\tau_j$.  Thus, among schedules in the collection $\{X_j\}_{j=0}^{k-1}$, the schedule $X_{l^*}$ with best consistency is such that
\[
l^*=\argmin_{j \in [0,k-1]} \alpha_j, \ \text{where} \ \alpha_j=\max_{i\in [0,k-1]} \frac{\tau_i}{\ell(X_j,\tau_i)}.
\] 
From Proposition~\ref{lem:4-robust}, it follows that the above algorithm yields a 4-robust schedule of optimal consistency. 
\end{proof}

\begin{figure*}[ht]
\includegraphics[width=\textwidth]{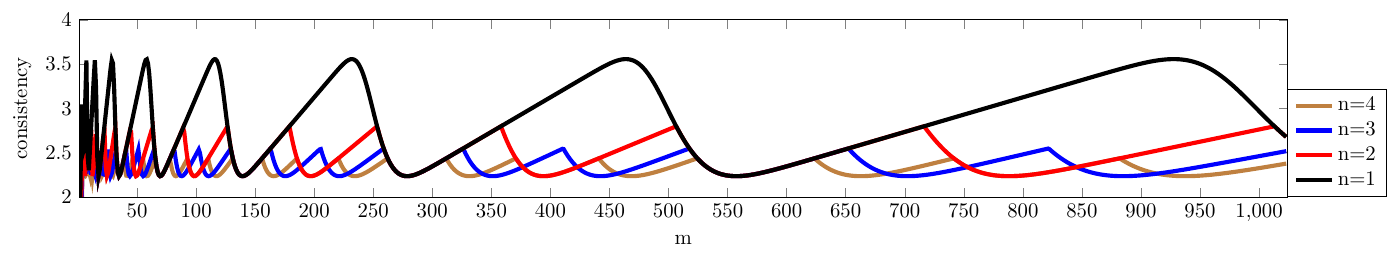}
\caption{Plot of the consistency of {\sc Sel}$_n$ with advice a truncated normal distribution as function of the mean {\tt m}.}    
\label{fig:normal.fixedv}
\end{figure*}

\begin{figure*}[ht]
\includegraphics[width=\textwidth]{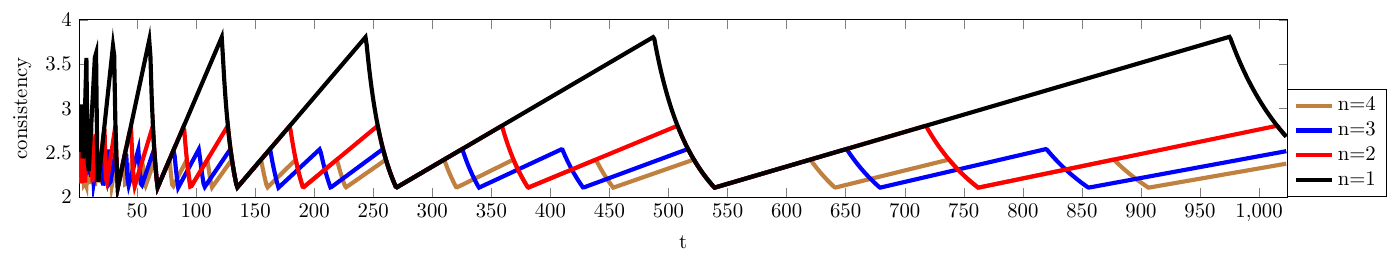}
\caption{Plot of the consistency of {\sc Sel}$_n$ with advice a uniform distribution in $[0.95t,1.05t]$, as function of $t$.}    
\label{fig:uniform}
\end{figure*}

The complexity of this algorithm is $O(k^2)$, since each $\alpha_j$ can be computed in time $O(k)$. However, we can reduce the complexity to $O(k\log k)$, using an argument that will also be useful in the proof of Corollary~\ref{cor:multiple.bounds}, and which is illustrated in Figure~\ref{fig:snapshot}. Consider the circle of unit perimeter, with an arbitrary point $O$ fixed, and let $p_0, \ldots ,p_{k-1}$ be points in the circle, such that the clockwise (arc-length) distance between $O$ and $p_j$ is $d_c(O,p_j)=\delta_j$. 
Given $i\in [0,k-1]$, define $D_i$ as the clockwise distance between the consecutive points $p_{(i-1) \bmod k}$ and $p_i$
in the circle. We show the following:

\begin{lemma}
For $j\in [0,k-1]$, let $X_j$ denote the schedule $(2^{i+\delta_j})_{i \in \mathbb{Z}}$. Then $X_j$ has consistency at most $2^{2-D_j}$.
\label{lem:circle}
\end{lemma}

\begin{proof}
From construction, it readily follows that the maximum consistency $c(X_j,P)$ is attained when the interruption occurs at time $\tau_{(j-1) \bmod k} \in P$.
We consider two cases. If $j \neq 0$, then $(j-1) \bmod k =j-1$, and by definition, we have that $\ell(X_j,\tau_{j-1})=2^{i_{j-1}-2+\delta_j}$. Therefore, 
\[
c(X_j,P) \leq \frac{2^{i_{j-1}+\delta_{j-1}}}{2^{i_{j-1}-2+\delta_j}}=2^{2-D_j}.
\]
For the second case, suppose that $j=0$, thus $(j-1) \bmod k=k-1$. Then, $\ell(X_0,\tau_{k-1})=2^{i_{k-1}+\delta_0-1}$. 
Therefore,
\[
c(X_0,P) \leq \frac{2^{i_{k-1}+\delta_{k-1}}}
{2^{i_{k-1}+\delta_0-1}} =2^{1+(\delta_{k-1}-\delta_0)}=2^{2-D_0},
\]
which concludes the proof. 
\end{proof}

Lemma~\ref{lem:circle} implies that we can improve the complexity of the algorithm of Theorem~\ref{thm:multiple.optimal}, by finding the index $j$ for which $D_{j}$ is minimized. This can be accomplished in time $O(k \log k)$ by sorting.

An interesting question is finding the exact value of the worst-case consistency of optimal schedules. The following corollary answers this question.

\begin{corollary}
The schedule of Theorem~\ref{thm:multiple.optimal} has consistency at most $2^{2-\frac{1}{k}}$, where $k$ is the size of $P$. Furthermore, this bound is tight, in that there exists a prediction $P$ such that every 4-robust schedule has consistency at least $2^{2-\frac{1}{k}}$.
\label{cor:multiple.bounds}
\end{corollary}

\begin{proof}
From Lemma~\ref{lem:circle}, since the smallest clockwise distance between two consecutive points in the unit circle cannot exceed $1/k$, we obtain that the worst-case consistency is at most $2^{2-\frac{1}{k}}$. This is tight, if the $k$ points are equidistantly distributed, i.e., if $\delta_j=j/k$, for every $j \in [0,k-1]$.
\end{proof}

Note that Corollary~\ref{cor:multiple.bounds} subsumes the known results for the extreme cases $k=1$ (for which the consistency is equal to 2~\cite{DBLP:journals/jair/AngelopoulosK23}), and $k\to \infty$ (for which the consistency reduces to the worst-case acceleration ratio~\cite{RZ.1991.composing}).

\section{Experimental evaluation}
\label{sec:experiments}

We present an experimental evaluation of our schedules, for both advice settings. We report the main findings, and we refer to the Appendix for additional results and discussion.

\subsection{Distributional advice}
\label{subsec:exp.distr}

We evaluate our algorithm of Theorem~\ref{thm:upper.dist}, to which we refer as {\sc Sel}$_n$, and recall that this algorithm selects the best schedule in class $S_n$, for some given $n \in \mathbb{N}$. We first consider, as distributional advice $\mu$, a {\em normal} distribution that is truncated at zero, with mean {\tt m}, and standard deviation $\sigma$. 

Figure~\ref{fig:normal.fixedv} depicts the experimental performance of {\sc Sel}$_n$, as function of {\tt m}, for $n \in [1,4]$, $\sigma=0.05 {\tt m}$ and ${\tt m} \in [1,1024]$. Note that for all sufficiently large {\tt m}, e.g., {\tt m}$\geq 100$, the expected interruption time is extremely close to {\tt m}.  We observe that the consistency of {\sc Sel}$_n$ improves as $n$ increases, by design of the algorithm. In particular, for $n=4$, the empirical consistency is 2.51. This value is, as expected, below the anticipated worst-case bound of Theorem~\ref{thm:upper.dist}
   (namely, $16(2^{0.25}-1)\approx 3.03$) and above the lower bound of 2, which applies to single, deterministic predictions. 

The shape of the consistency is a saw-like function of the mean, which is explained by the fact that the consistency (as the acceleration ratio) has transitions at ``critical'' times, i.e., right after the completion of a contract. As $n$ increases, we note that the transitions become more smooth. This is because the number of candidate schedules in $S_n$ becomes larger, hence the probability that the interruption is critical for some schedule in $S_n$ decreases. This also explains why the number of peaks increases as function of $n$.

 
We also report results for a {\em uniform} distributional advice. Specifically, Figure~\ref{fig:uniform} depicts the consistency of {\sc Sel}$_n$ with advice chosen according to $U[0.95t, 1.05t]$, as function of $t$. For $n=4$, we observe a similar empirical consistency value as in the case of the normal distribution, namely 2.44. 
The shape of the plot has sharper transitions, relatively to the normal distribution, which can be explained by the fact that the latter spreads the random interruption over a larger interval than the uniform one, in this experimental setup. This also explains why the difference between the peaks and valleys is more pronounced in Figure~\ref{fig:uniform}
than in Figure~\ref{fig:normal.fixedv}.

\subsection{Multiple advice}

We evaluate our algorithm from Section~\ref{sec:multiple}, to which we refer as {\sc Mult}$_k$, where $k$ is the cardinality of the prediction set $P$. Specifically, for every $k \in [1,10]$, we generate $P$ as $k$ values chosen independently and uniformly at random in the interval $[1,1024]$. For each such prediction set, we compute the worst-case consistency, as evaluated by~\eqref{eq:cons.mult}. We compute, in addition, the average-case consistency, namely the average of the ratios $\tau/\ell({\text{\sc Mult}}_k,\tau)$, where $\tau$ is chosen uniformly at random from $P$. The latter is a much more relaxed performance measure that treats each of the $k$ predictions in $P$ as equally likely. We repeat this experiment 1000 times, and we report the average of the corresponding ratios. 

We observe that for all values of $k$, both the worst-case and the average consistency are below the upper bound of 
$2^{2-\frac{1}{k}}$, which confirms the result of Corollary~\ref{cor:multiple.bounds}. The ratios, as expected, are increasing functions of $k$, since the larger the size of $P$, the more noisy the quality of the prediction.

\begin{figure}[t!]
    \centering
    \includegraphics[width=9cm]{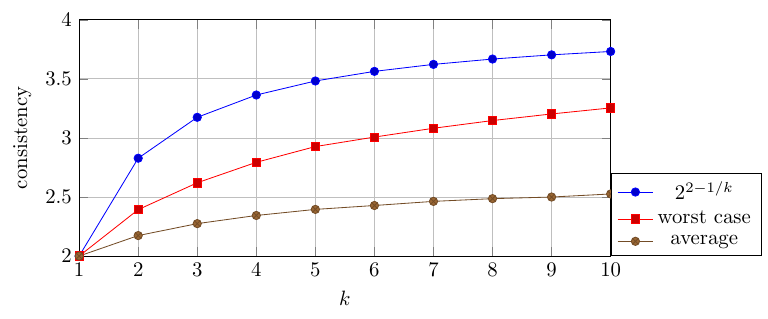}
    \caption{Experimental evaluation of {\sc MULT}$_k$.}
    \label{fig:multiple}
\end{figure}

\section{Conclusion}

We studied a classic optimization problem related to bounded-resource reasoning, namely contract scheduling, in novel learning-augmented settings that capture distributional and multiple predictions. In both settings, we gave analytically optimal schedules that simultaneously optimize the consistency and the robustness. As discussed in Section~\ref{sec:introduction}, contract scheduling has been studied in a variety of settings, including multiple instances that must be solved concurrently, in a single or multiple processors. Future work will address these more complex variants under similar learning-enhanced models. It will be also interesting to perform a multi-objective analysis, in the multiple-advice model, that addresses the average expected consistency or trade-offs between the worst-case and the average-case metrics. The techniques of Section~\ref{sec:multiple} can be applicable, by showing first tradeoffs between average and worst-case distances between points on the unit circle.

Another direction for future work is searching for a hidden target, with distributional or multiple predictions about its position in the environment. This setting has applications in robotic search and exploration, and has been studied with single, or no prediction e.g.,~\cite{DBLP:conf/aaai/EberleLMNS22,sung2019competitive}. In particular, the techniques we developed in this work will be very useful in the context of the {\em line} and {\em star} search environments~\cite{jaillet:online}, given the connections between contract scheduling and competitive search~\cite{steins,spyros:ijcai15}. Last, our work is the first to bring attention to the fact that single predictions may be fragile for certain problems, a finding that can have implications in other domains and applications such as learning-augmented online conversion problems~\cite{sun2021pareto}. 

\section{Acknowledgements}

This work was partially funded by the project PREDICTIONS, grant ANR-23-CE48-0010 and the project ALGORIDAM, grant ANR-19-CE48-0016 from the French National Research Agency (ANR), as well as by Polish National Science Centre grant 2022/45/B/ST6/00559.

\bibliographystyle{named}
\bibliography{targets-arxiv}

\clearpage
\newpage

\appendix

{\large \bf Appendix}

\section{Omitted proofs and proof details}

\noindent
\begin{proof}[Proof of Proposition~\ref{lem:4-robust}] 
For the first part of the proposition, we know from~\cite{RZ.1991.composing} that the robustness of a schedule (i.e., its acceleration ratio) is maximized for interruptions that occur infinitesimally prior to the completion time of a contract in the schedule. Thus, an equivalent expression of the robustness of a schedule $X$ is given by
\begin{equation}
r(X)=\sup_i\frac{\sum_{j=-\infty}^i x_j}{x_{i-1}}.
\label{eq:robust,equiv}
\end{equation}
It is then straightforward to confirm that any schedule $X(\lambda)$ is 4-robust, since it is a scaled variant of the doubling schedule, which is 4-robust. For the second part of the proposition, the result follows from the study of the linear recurrence inequality $\sum_{j\leq i} x_j \leq 4 x_{i-1}$, and in particular Theorem 9.4 in~\cite{searchgames}, with $M=2$. The result is described in terms of a search problem, but the same recurrence inequality applies, and hence the same conclusion.
\end{proof}

\bigskip



\input{distrib-LB-appendix}

\begin{proof}[Proof of Proposition~\ref{prop:single}]
Let $X$ denote a 4-robust, 2-consistent schedule $X$. Then for any given prediction $\tau$, from~\cite{DBLP:journals/jair/AngelopoulosK23} we know that 
$X$ must complete a contract $C$ of length $\tau/2$ at time $\tau$, as $\tau \to \infty$. Consider now the actual interruption $T=\tau-\epsilon$, i.e., the interruption occurs right before $C$ terminates. Then
\[
\ell(X,T)=\frac{\tau}{4} =\frac{T+\epsilon}{4}.
\]
\end{proof}

\bigskip

\input{smoothness-appendix}



\section{Additional experimental results}

We report additional experimental results concerning the distributional advice model. 

\medskip

\subsection{Time horizon}
We repeated the experiments of Section~\ref{subsec:exp.distr}
for larger time horizons, namely $[1,20000]$, shown in Figure~\ref{fig:time.normal} and Figure~\ref{fig:time.uniform}. The plots have the same overall shape as those reported in the main paper, which is expected, considering that the schedule follows a scaled doubling strategy.

\medskip
\subsection{Fixed variance} 
\label{app:fixed}
We considered again the setting of Section~\ref{subsec:exp.distr}, with the difference that the variance of the distributional advice is now fixed, instead of being a function of time (or the mean {\tt m}). This captures the situation in which the window of time in which we expect the interruption to occur is prespecified. Specifically, Figure~\ref{fig:normal.fixed.appendix} depicts the consistency of {\sc Sel}$_n$ given a normal distributional advice of mean {\tt m}, and standard deviation $\sigma=10$.

We observe that as the time increases, the transitions of the plot become sharper, relatively to the setting in which the variance is fixed. This is because, as the time increases, the variance becomes less significant in relation to the mean of the distributional advice. This implies, in turn, that the advice has decreasing spread as function of time, and becomes more and more close to a single-valued prediction for large $t$. This also explains why for the same value of $n$, the consistency of {\sc Sel}$_n$ is worse in this setting than in the setting of Figure~\ref{fig:normal.fixedv}, in the main paper.
The values reported for very small values of {\tt m} are outlier, and are due to the high variance (relative to the small values of {\tt m}).

\begin{figure*}[t!]
\includegraphics[width=\textwidth]{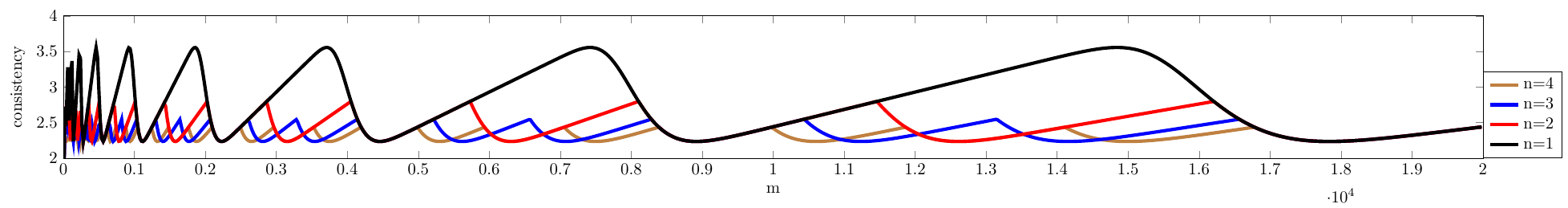}
\caption{Plot of the consistency of {\sc Sel}$_n$ with time horizon $[1,20000]$, for normal distributional advice as in the main paper.}    
\label{fig:time.normal}
\end{figure*}

\begin{figure*}[t!]
\includegraphics[width=\textwidth]{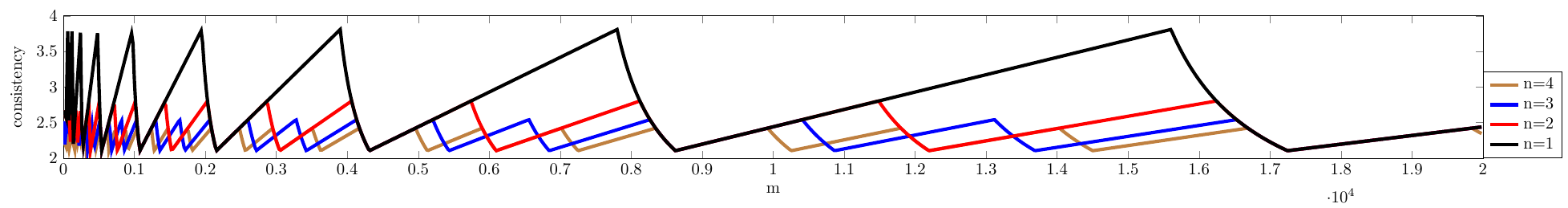}
\caption{Plot of the consistency of {\sc Sel}$_n$ with time horizon $[1,20000]$, for uniform distributional advice as in the main paper.}    
\label{fig:time.uniform}
\end{figure*}

\begin{figure*}[t!]
\includegraphics[width=\textwidth]{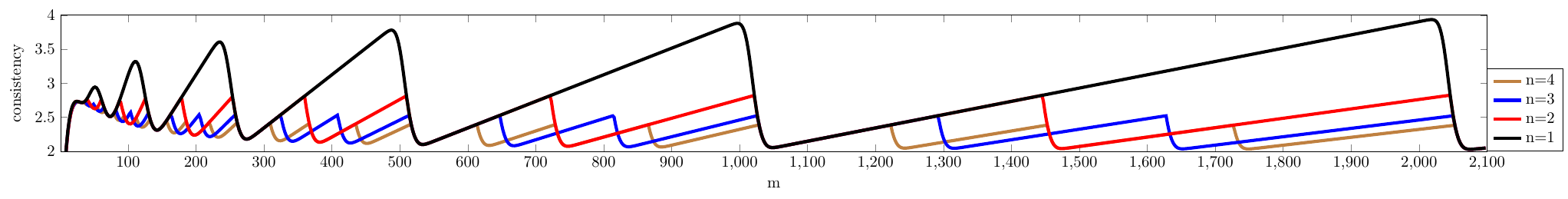}
\caption{Plot of the consistency of {\sc Sel}$_n$  for normal distributional advice with mean {\tt m} and fixed standard deviation $\sigma=10$.}    
\label{fig:normal.fixed.appendix}
\end{figure*}

\begin{figure*}[t!]
\includegraphics[width=\textwidth]{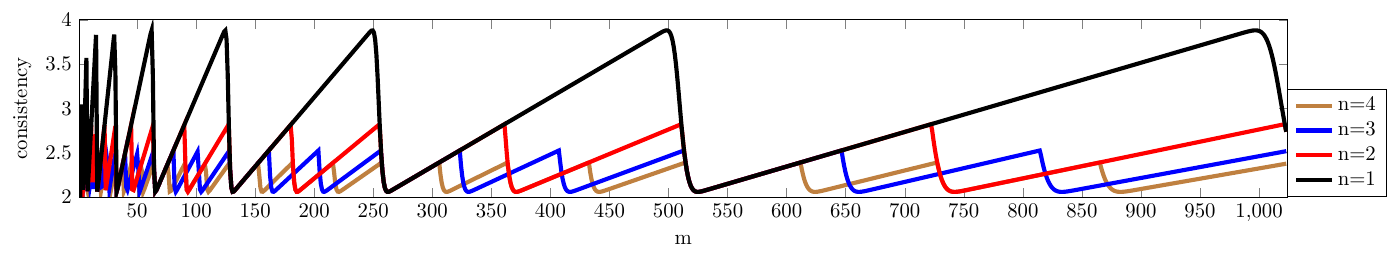}
\caption{Plot of the consistency of {\sc Sel}$_n$  for normal distributional advice with mean {\tt m} and standard deviation $\sigma=0.01 {\tt m}$.}    
\label{fig:normal.variable.001}
\end{figure*}

\medskip
\subsection{Time-dependent variance} 

We did experiments with a distributional advice whose mean is variable, and function of time, and whose variance is larger or smaller than the values of the main paper. Figure~\ref{fig:normal.variable.001} and Figure~\ref{fig:normal.variable.02} depict the experimental consistency with normal distributional advice for a setting as in the main paper, and standard deviation $\sigma=0.01 {\tt m}$ and $\sigma=0.2 {\tt m}$, respectively. For small $\sigma$, we observe sharper transitions, and a larger gap between the peaks and valleys, which can be explained along the same lines as in the discussion of Section~\ref{app:fixed}. For large $\sigma$, we observe that the consistency increases to values close to 2.7, and that the plots become relatively ``flat'', with much smoother transitions. This is because, with large variance, the setting becomes highly random, and the prediction less crucial or useful, which also explains the much smaller gap between the various values of $n$. 

Similar observations and conclusions can be drawn for uniform advice, as depicted in Figure~\ref{fig:uniform.variable.01}.

\medskip
\subsection{Error analysis}

We conducted experiments to test the resilience of the algorithm {\sc Sel}$_n$ to errors in the distributional advice. In particular, we considered the setting in which the advice is a normal distribution with mean $m$ and standard deviation fixed to $\sigma=25$, whereas the interruption time is drawn according to a normal distribution with {\em actual}, but unknown mean $m'$ (in general, $m'\neq m$) and standard deviation equal to 25. We choose $n=16$.

In Figure~\ref{fig:error.500.small}, the blue curve depicts the performance of our schedule as function of the actual mean $m'$ (and thus, as function of the error, defined as $\eta=|m'-m|$). 
In contrast, the red curve depicts the consistency of the schedule with error-free advice. We observe that as long as the error is relatively small, the schedule remains fairly robust to prediction errors, and shows smooth degradation. We also observe an asymmetry on the performance, depending on whether the error is ``positive'', i.e., $m'>m$ or ``negative'', i.e., $m'<m$). This is expected, since a positive error does not have as critical an effect as negative error: positive error implies that a large contract in the schedule will still likely terminate, whereas negative error implies that the schedule is running a large contract which likely will {\em not} terminate by the interruption. 

Figure~\ref{fig:error.500.large} depicts the same plots for much larger values of error. We note that the performance plot exhibits, informally, a certain periodicity, which can be explained by the nature of the problem. Namely, even if the error becomes large, it is possible that the advice may have a beneficial effect to the schedule, if the latter happens to complete a large contract close to the actual, expected interruption time. This explains why the blue and red curves meet not only for $m'=m$, but also for other values of $m'$. This finding is also an unavoidable situation: {\em any} schedule will perform very well for some outlier values of large error, which implies that it is impossible to find a schedule that strictly dominates all other schedules across any error domain. 

Figure~\ref{fig:error.700.large} depicts the same plot for normal distributional advice with $m=700$. We observe similar performance and findings as discussed above.

\begin{figure*}[t!]
\includegraphics[width=\textwidth]{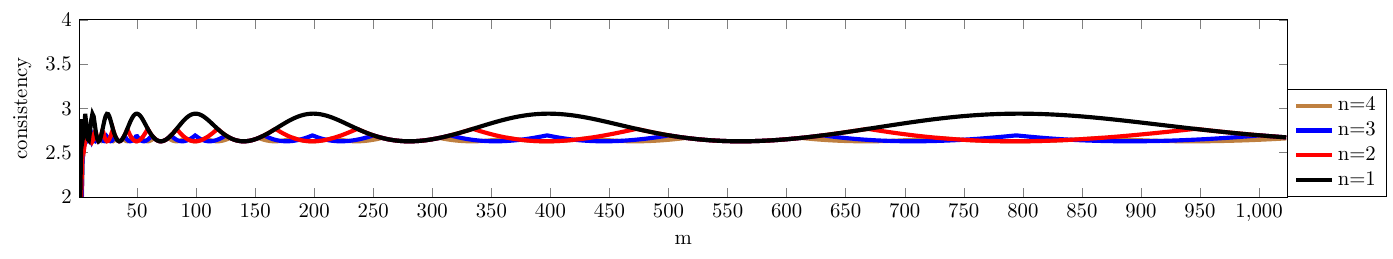}
\caption{Plot of the consistency of {\sc Sel}$_n$  for normal distributional advice with mean {\tt m} and standard deviation $\sigma=0.2 {\tt m}$.}    
\label{fig:normal.variable.02}
\end{figure*}

\begin{figure*}[t!]
\includegraphics[width=\textwidth]{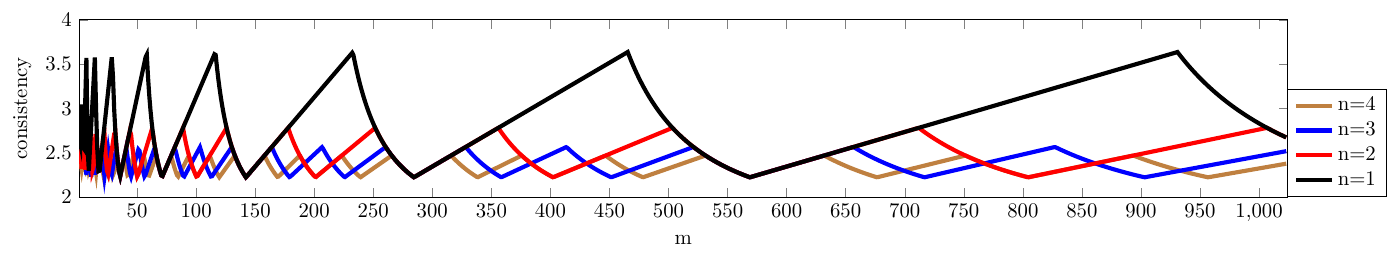}
\caption{Plot of the consistency of {\sc Sel}$_n$  for uniform distributional advice $U[0.9t,1.1t]$, as function of time $t$.}
\label{fig:uniform.variable.01}
\end{figure*}

\begin{figure*}[t!]
\includegraphics[width=\textwidth]{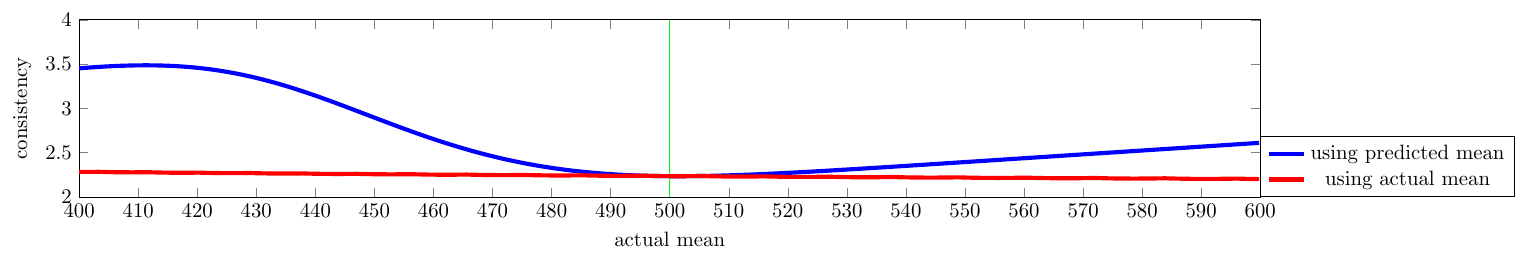}
\caption{Plot of the consistency of {\sc Sel}$_{16}$  
with normal distributional advice of mean $m=500$,
as function of the actual mean $m'$.}
\label{fig:error.500.small}
\end{figure*}

\begin{figure*}[t!]
\includegraphics[width=\textwidth]{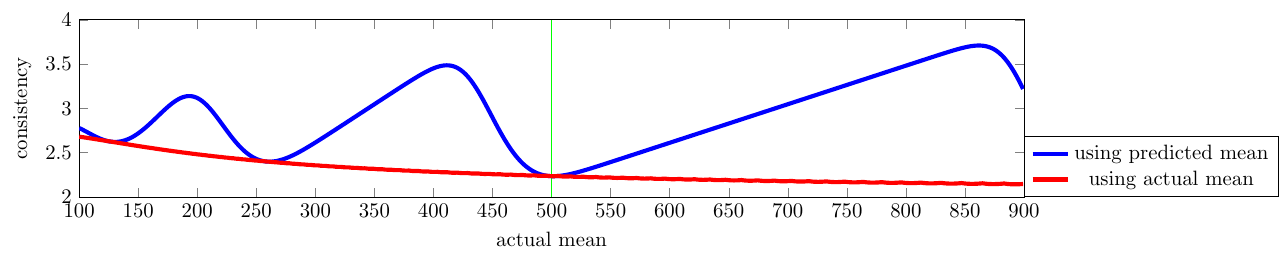}
\caption{Plot of the consistency of {\sc Sel}$_{16}$  
with normal distributional advice of mean $m=500$,
as function of the actual mean $m'$ (larger range of error).}
\label{fig:error.500.large}
\end{figure*}

\begin{figure*}[t!]
\includegraphics[width=\textwidth]{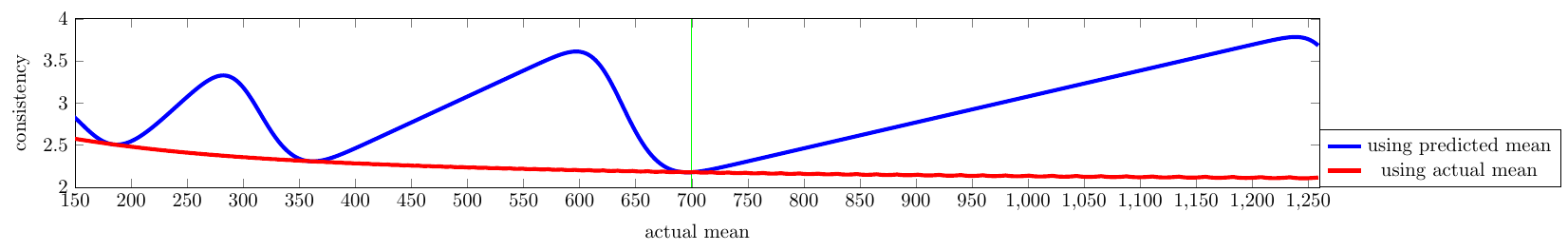}
\caption{Plot of the consistency of {\sc Sel}$_{16}$  
with normal distributional advice of mean $m=700$,
as function of the actual mean $m'$ (larger range of error).}
\label{fig:error.700.large}
\end{figure*}

\end{document}

%% file: distrib-LB.tex
\newcommand{\expmul}[2][\mu]{\ensuremath{\mathbb E_{z\sim#1}[\ell(#2,z)]}\xspace}
\newcommand{\expmu}[1][\mu]{\ensuremath{\mathbb E_{z\sim#1}[z]}\xspace}

\begin{theorem}
For any $n \in \mathbb N$, there exists a distributional prediction for which every collection $C_n$ that consists of $n$ 4-robust schedules cannot contain a schedule of consistency smaller than $4n \cdot (2^{1/n}-1)$.
\label{thm:lower.dist}
\end{theorem}

\begin{proof}
Using Proposition~\ref{lem:4-robust}, let us denote, without loss of generality, the 4-robust schedules in $C_n$ by $X(\lambda_k) = (2^{i-\lambda_k})_{i=-\infty}^{\infty}$ for $k\in[0,n-1]$ and $0\leq\lambda_k<\lambda_{k+1}< 1$.

We will define a distribution $\mu$ that has an $n$-point support, and where the $k$-th point occurs at time $2^{2-\lambda_k-\varepsilon}$ with probability $p_k$, for $\varepsilon>0$ arbitrarily small. For $k<n-1$, we define  $p_k = 2^{\lambda_{k+1}-\lambda_0} -2^{\lambda_{k}-\lambda_0}$, and define $p_{n-1} = 2-2^{\lambda_{n-1}-\lambda_0}$. Let $P_k=\sum_{j\leq k} p_j$, so $P_k=2^{\lambda_{k+1}-\lambda_0}-1$.

We have that  $\E_{z\sim \mu} [\ell(X(\lambda_0),z)]= 2^{-\lambda_0}$ as the support of $\mu$ is in the interval $[2^{1-\lambda_0},2^{2-\lambda_0}]$. Therefore, if an interruption occurs at any point in the support of $\mu$, the last completed contract of $X(\lambda_0)$ is of length $2^{-\lambda_0}$.

Similarly, for any $k>0$, we have:
\begin{align}
\expmul{&X_{\lambda_k}} =
P_{k-1}2^{1-\lambda_k}+(1-P_{k-1})2^{-\lambda_k}\nonumber\\ &= 2^{1-\lambda_0}-2^{1-\lambda_k}+2^{1-\lambda_k}-2^{-\lambda_0} 
~~ = 2^{-\lambda_0}.
\label{eq:lb1}
\end{align}

The expected value of $\mu$ is (defining $\lambda_{n}=\lambda_0+1$):
\begin{align}
\expmu &=  \sum_{j<n} (2^{\lambda_{j+1}-\lambda_0} - 2^{\lambda_{j}-\lambda_0})\cdot 2^{2-\lambda_j} \nonumber\\
&= 2^{-\lambda_0} \cdot \sum_{j<n} (2^{2+\lambda_{j+1}-\lambda_j}-2^2) \nonumber\\
&= 2^{2-\lambda_0} \cdot \Big(-n+\sum_{j<n}  2^{\lambda_{j+1}-\lambda_j}\Big).
\label{eq:lb2}
\end{align}
From~\eqref{eq:lb1} and~\eqref{eq:lb2}, the consistency of any $X(\lambda_k)$ is: 
$$ c(X(\lambda_k),\mu)= 4\cdot \Big(- n + \sum_{j<n}  2^{\lambda_{j+1}-\lambda_j}\Big).$$

$c(X(\lambda_k),\mu)$ is minimized when the $\lambda_j$ are chosen so as to minimize the sum term, while respecting the constraint that the sum of all the exponents equals $\lambda_n-\lambda_0=1$ and each exponent belongs to $[0,1)$. By the convexity of the function $x\mapsto2^x$, this occurs when all terms are equal to $2^{1/n}$, hence
$$ c(X(\lambda_k),\mu) \geq 4\cdot\left(-n+n\cdot 2^{1/n}\right) = 4n\cdot(2^{1/n}-1).\qedhere$$
\end{proof}

By allowing $n \to \infty$, Theorem~\ref{thm:lower.dist} shows that no 4-robust schedule can have consistency better than $4\ln 2$. Note that the theorem relies on a probability distribution with $n$ mass points. However, we  can obtain the same lower bound by relying to a prediction given as a continuous distribution.

\begin{theorem}[Appendix]
For any $D>0$, there is a continuous distribution $\mu_D$ over the interval $[D,2D]$ such that, for every 4-robust schedule $X$,  $c(X,\mu_D) \geq 4\ln2.$
\label{thm:lower.dist.cont}
\end{theorem}

%% file: smoothness.tex
\newcommand{\rp}{\ensuremath{\mathbb R^+}\xspace}
\newcommand{\emd}{\ensuremath{\mathrm{EMD}}\xspace}

In contrast, we will show that this is not the case for the worst-case distributional prediction of Theorem~\ref{thm:lower.dist.cont}. Specifically, in Theorem~\ref{lem:smoothness}, we will show that the performance of the schedule degrades smoothly as a function of the prediction error. For some intuition behind the proof, we will use the fact that worst-case distributional predictions require that a large time interval contributes to the schedule's expected profit, and that the distribution is rather ``balanced'' over that interval. Therefore, if an adversary were to change substantially the performance of the schedule, by altering the predicted distribution, it would have to shift a large amount of probability mass over a long time span. This motivates the choice of the well-known Earth Mover's Distance (EMD) as a metric of the prediction error. We will show that a small EMD error has a likewise small effect on the performance.

Given two probability distributions $\mu$ and $\mu'$ over $\rp$, the Earth Mover's Distance (\emd)~\cite{rubner1998metric}  intuitively represents the minimum cost incurred in order to obtain $\mu'$ from $\mu$, where moving an infinitesimal probability mass costs its amount multiplied by the distance moved. Formally, let
$\Psi(\mu,\mu')\subset(\rp\times \rp)\rightarrow \rp$ be the set of functions
$\psi$ such that $\forall x\in \rp,~ \int_{\rp} \psi(x,y) \mathrm dy = \mu(x)$
and $\forall y\in \rp,~ \int_{\rp} \psi(x,y) \mathrm dx = \mu'(y)$. Here, $\psi(x,y)$ describes the amount of mass that is moved from point $x$ in $\mu$ to 
point $y$ in $\mu'$. The distance $\emd(\mu,\mu')$ is then defined as
$$\emd(\mu,\mu') = \inf_{\psi\in\Psi(\mu,\mu')} \int_{\rp} \int_{\rp} |x-y| \psi(x,y) \mathrm dx \mathrm dy.$$ 

\begin{theorem}
\label{lem:smoothness} Given $D\in\rp$, let $\mu_D$ be the probability distribution of Theorem~\ref{thm:lower.dist.cont}.
Then, for, any probability distribution $\mu'$ such that
$\emd(\mu_D,\mu')\leq\eta$, for sufficiently small $\eta$, any 4-robust schedule $X$  with prediction $\mu_D$ satisfies
\[
\E_{z\sim \mu'}[\ell(X,z)] \geq \frac{\E_{z\sim \mu'}[z]}{4\ln 2+O(\sqrt{\eta/D})}.
\]
\end{theorem}

\begin{proof}[Proof Sketch ] We give the main elements of the proof, and defer some technical parts to the appendix.

Recall that the distribution $\mu = \mu_D$ is defined over $[D,2D]$ by the density function $x\mapsto 2D/x^2$. This density function takes values in the interval $[1/(2D),2/D]$.
From Proposition~\ref{lem:4-robust}, we know that any
4-robust schedule $X$ is of the form 
$X(\lambda) = 2^{i-\lambda}$, for some $\lambda\in [0,1)$.
For given $\eta>0$,  let $\mu'$ be any probability distribution such that $\emd(\mu,\mu')\leq \eta.$

By the definition of $\emd$, let $\psi\in\Psi(\mu,\mu')$ be a function such that 
$$\emd(\mu,\mu') = \int_{\rp} \int_{\rp} |x-y| \psi(x,y) \mathrm dx \mathrm dy.$$ 

We decompose $\psi$ in two functions $\psi^+$ and $\psi^-$
satisfying the following conditions. For $x<y$, $\psi^+(x,y)=\psi(x,y)$ and
$\psi^-(x,y)=0$. For $x>y$, $\psi^-(x,y)=\psi(x,y)$ and $\psi^+(x,y)=0$. Last, we require that $\psi^+(x,x) = \int_{z\leq x} \psi(x,z)\mathrm dz$, and $\psi^-(x,x) =
\int_{z\geq x} \psi(x,z)\mathrm dz$. Intuitively, $\psi^+$ denotes the probability mass moved towards higher values,
and, likewise, $\psi^-$ denotes the probability mass moved
towards lower values. Let $\mu^+$ be the probability obtained after applying
$\psi^+$; namely, we define $\mu^+(y) = \int_{\rp} \psi^+(x,y) \mathrm dx$.  Note that
$$\eta \geq \emd(\mu,\mu') =\emd(\mu,\mu^+)+\emd(\mu^+,\mu').$$

In the first part of the proof (details in the appendix), we focus on the effect of $\psi^+$ and bound the expected profit of $X(\lambda)$ under the distribution $\mu^+$, as well as the expectation of the distribution $\mu^+$. The intuition here is that the expected profit of $X(\lambda)$ cannot decrease, and that the growth of the expected value of the distribution is bounded by the EMD. We show 

\begin{equation}  
\label{eq:emdemuxpMAIN}
\expmul[\mu^+]{X(\lambda)}\geq \expmul{X(\lambda)},
\end{equation}
\begin{align} 
\expmu[\mu^+] &\leq \expmu + \eta.
\label{eq:emdemupMAIN}
\end{align}

In the next part of the proof (details in the appendix), we focus on the effect of $\psi^-$. Here, the expected value of the distribution cannot increase, and the expected profit of $X(\lambda)$ decreases by an amount which can be bounded using $\eta$ and $D$. This last bound is the most technical step of the proof and uses the fact that $\min_{x\in[D,2D]}\mu(x) \leq 4\cdot \max_{x\in[D,2D]}\mu(x)$. We show

\begin{equation}
\label{eq:emdemuprimeMAIN}    
\expmu[\mu'] \leq \expmu[\mu^+],
\end{equation}
\begin{align}
\expmul[\mu']{X(\lambda)} &\geq 
\expmul[\mu^+]{X(\lambda)}-8\sqrt{{2\eta D}}.
\label{eq:emdexmuprime2MAIN}
\end{align}

To conclude the proof, note that $\expmu=2D\ln 2$
and $\expmul{X(\lambda)}=D/2$, from the proof of Theorem~\ref{thm:lower.dist.cont}. 
Combing Equations~\eqref{eq:emdemuxpMAIN}, \eqref{eq:emdemupMAIN}, \eqref{eq:emdemuprimeMAIN} and \eqref{eq:emdexmuprime2MAIN}, we arrive at the desired result, for $\eta<D/512$ (to ensure there is no division by zero). 
See Appendix for details.
\end{proof}

%% file: circle-figure.tex
\newcommand{\snap}{

\tikzset{every picture/.style={line width=0.75pt}} 

\begin{tikzpicture}[x=0.75pt,y=0.75pt,yscale=-1,xscale=1]

\draw   (210,174) .. controls (210,111.04) and (261.04,60) .. (324,60) .. controls (386.96,60) and (438,111.04) .. (438,174) .. controls (438,236.96) and (386.96,288) .. (324,288) .. controls (261.04,288) and (210,236.96) .. (210,174) -- cycle ;
\draw  [fill={rgb, 255:red, 0; green, 0; blue, 0 }  ,fill opacity=1 ] (263,273.5) .. controls (263,271.57) and (264.57,270) .. (266.5,270) .. controls (268.43,270) and (270,271.57) .. (270,273.5) .. controls (270,275.43) and (268.43,277) .. (266.5,277) .. controls (264.57,277) and (263,275.43) .. (263,273.5) -- cycle ;
\draw  [fill={rgb, 255:red, 0; green, 0; blue, 0 }  ,fill opacity=1 ] (392,263.5) .. controls (392,261.57) and (393.57,260) .. (395.5,260) .. controls (397.43,260) and (399,261.57) .. (399,263.5) .. controls (399,265.43) and (397.43,267) .. (395.5,267) .. controls (393.57,267) and (392,265.43) .. (392,263.5) -- cycle ;
\draw  [fill={rgb, 255:red, 0; green, 0; blue, 0 }  ,fill opacity=1 ] (429,141.5) .. controls (429,139.57) and (430.57,138) .. (432.5,138) .. controls (434.43,138) and (436,139.57) .. (436,141.5) .. controls (436,143.43) and (434.43,145) .. (432.5,145) .. controls (430.57,145) and (429,143.43) .. (429,141.5) -- cycle ;
\draw  [fill={rgb, 255:red, 0; green, 0; blue, 0 }  ,fill opacity=1 ] (391,84.5) .. controls (391,82.57) and (392.57,81) .. (394.5,81) .. controls (396.43,81) and (398,82.57) .. (398,84.5) .. controls (398,86.43) and (396.43,88) .. (394.5,88) .. controls (392.57,88) and (391,86.43) .. (391,84.5) -- cycle ;
\draw  [fill={rgb, 255:red, 0; green, 0; blue, 0 }  ,fill opacity=1 ] (226,112.5) .. controls (226,110.57) and (227.57,109) .. (229.5,109) .. controls (231.43,109) and (233,110.57) .. (233,112.5) .. controls (233,114.43) and (231.43,116) .. (229.5,116) .. controls (227.57,116) and (226,114.43) .. (226,112.5) -- cycle ;
\draw  [fill={rgb, 255:red, 0; green, 0; blue, 0 }  ,fill opacity=1 ] (324,60) .. controls (324,58.07) and (325.57,56.5) .. (327.5,56.5) .. controls (329.43,56.5) and (331,58.07) .. (331,60) .. controls (331,61.93) and (329.43,63.5) .. (327.5,63.5) .. controls (325.57,63.5) and (324,61.93) .. (324,60) -- cycle ;
\draw [line width=1.5]    (327.5,54.5) .. controls (381.89,37.58) and (518.79,144.42) .. (406.22,266.66) ;
\draw [shift={(404.5,268.5)}, rotate = 313.24] [color={rgb, 255:red, 0; green, 0; blue, 0 }  ][line width=1.5]    (14.21,-4.28) .. controls (9.04,-1.82) and (4.3,-0.39) .. (0,0) .. controls (4.3,0.39) and (9.04,1.82) .. (14.21,4.28)   ;
\draw [line width=1.5]    (399.17,272.33) .. controls (351.67,311.28) and (287.92,294.28) .. (267.95,287.34) ;
\draw [shift={(265.17,286.33)}, rotate = 20.56] [color={rgb, 255:red, 0; green, 0; blue, 0 }  ][line width=1.5]    (14.21,-4.28) .. controls (9.04,-1.82) and (4.3,-0.39) .. (0,0) .. controls (4.3,0.39) and (9.04,1.82) .. (14.21,4.28)   ;

\draw (308,68) node [anchor=north west][inner sep=0.75pt]  [font=\LARGE] [align=left] {O};
\draw (361,84) node [anchor=north west][inner sep=0.75pt]  [font=\LARGE] [align=left] {$\displaystyle p_{0}$};
\draw (396,130) node [anchor=north west][inner sep=0.75pt]  [font=\LARGE] [align=left] {$\displaystyle p_{1}$};
\draw (361,222) node [anchor=north west][inner sep=0.75pt]  [font=\LARGE] [align=left] {$\displaystyle p_{j}$};
\draw (204,253) node [anchor=north west][inner sep=0.75pt]  [font=\LARGE] [align=left] {$\displaystyle p_{j+1}$};
\draw (160,85) node [anchor=north west][inner sep=0.75pt]  [font=\LARGE] [align=left] {$\displaystyle p_{k-1}$};
\draw (456,121) node [anchor=north west][inner sep=0.75pt]  [font=\LARGE] [align=left] {$\displaystyle \delta _{j}$};
\draw (323,305) node [anchor=north west][inner sep=0.75pt]  [font=\LARGE] [align=left] {$\displaystyle D_{j+1}$};
\draw (162,151) node [anchor=north west][inner sep=0.75pt]  [font=\LARGE] [align=left] {$\displaystyle \dotsc $};
\draw (384,179) node [anchor=north west][inner sep=0.75pt]  [font=\LARGE] [align=left] {$\displaystyle \dotsc $};

\end{tikzpicture}

}

%% file: distrib-LB-appendix.tex
\begin{proof}[Proof of Theorem~\ref{thm:lower.dist.cont}]
From Proposition~\ref{prop:single}, let $X(\lambda)$ denote the 4-robust schedule, for some $\lambda\in[0,1)$. Consider any value $D>1$. We will define a distribution function $\mu_D$ over $[D,2D]$ such that $c(X(\lambda),\mu_D)\geq 4\ln 2$.

Specifically, let the density function of $\mu_D$ equal $f_D(x) = 2D/x^2$ on $[D,2D]$ and $0$
elsewhere. Note that $\int f_D=1$, so $\mu_D$ is indeed a distribution. We have
$$\expmu[\mu_D] = \int_{D}^{2D}\frac{2D}x dx = 2D\ln 2.$$

Let $k$ be such that $2^{k-\lambda}\in[D,2D)$.
The last completed contract of $X(\lambda)$ is of length $2^{k-\lambda-2}$ if an interruption occurs between $D$ and $2^{k-\lambda}$, and of length $ 2^{k-\lambda-1}$ if an interruption occurs between $2^{k-\lambda}$ and $2D$. So
\begin{align*}
&\expmul[\mu_D]{X(\lambda)} \\&= 2^{k-2-\lambda}\int_{D}^{2^{k-\lambda}} \frac{2D}{x^2} dx + 2^{k-1-\lambda}\int_{2^{k-\lambda}}^{2D} \frac{2D}{x^2} dx\\
&= 2^{k-2-\lambda}\left(2 - D\cdot2^{1-k+\lambda} + 2 (2D\cdot2^{\lambda-k}-1)\right)\\
&= 2^{k-2-\lambda} (D\cdot 2^{1-k+\lambda}) \\
&= D/2.
\end{align*}

Therefore, we have $c(X(\lambda),\mu_D) = 4\ln2$.
\end{proof}

\bigskip

%% file: smoothness-appendix.tex
\begin{proof}[Omitted details of the proof of Theorem~\ref{lem:smoothness}]
    
In the first part of the proof, we will focus on the effect of $\psi^+$. By construction of $\psi^+$, we have: 
$$\emd(\mu,\mu^+)=\int_{x\in\rp} \int_{y>x} (y-x) \psi(x,y) \mathrm dy
\mathrm dx.$$ 
For any $x$, we have $\mathbb P_{z\sim{\mu^+}}(x>z) \geq \mathbb P_{z\sim{\mu}}(x>z)$, therefore we obtain Equation~\eqref{eq:emdemuxpMAIN} from the main paper:

\begin{equation}  
\tag{\ref{eq:emdemuxpMAIN}}
\expmul[\mu^+]{X(\lambda)}\geq \expmul{X(\lambda)}.
\end{equation}

Moreover,
$$\expmu = \int_{\rp} x \cdot \mu(x)\mathrm dx = \int_{\rp}\int_{\rp} x\cdot \psi^+(x,y) \mathrm dy \mathrm dx,$$ 
hence, we obtain Equation~\eqref{eq:emdemupMAIN} from the main paper:
\begin{align} \nonumber
\expmu[\mu^+] &= \int_{\rp}\int_{\rp} y\cdot \psi^+(x,y) \mathrm dx \mathrm dy\\ 
\nonumber
&= \expmu + \emd(\mu,\mu^+)\\
&\leq \expmu + \eta.
\tag{\ref{eq:emdemupMAIN}}
\end{align}

In the next part of the proof, we focus on the effect of $\psi^-$. By similar arguments, we can see that, as 
$\mathbb P_{z\sim{\mu'}}(x>z) \leq \mathbb P_{z\sim{\mu^+}}(x>z)$, we derive Equation~\eqref{eq:emdemuprimeMAIN} from the main paper:
\begin{equation}
\tag{\ref{eq:emdemuprimeMAIN}}
\expmu[\mu']\leq \expmu[\mu^+].
\end{equation}

We now give a lower bound on
$\expmul[\mu']{X(\lambda)}$. We have
$$\eta\geq  \emd(\mu^+,\mu')=\int_{x\in\rp} \int_{y<x} (x-y) \psi(x,y) \mathrm dy \mathrm dx.$$ 

Let $k$ such that $2^{k-\lambda}\in[D,2D)$. As $\mu^+(x)=0$ for any $x<D$, 
the last completed contract of $X(\lambda)$ is of length $D$ if the interruption is in the interval $[2^{k-\lambda},2D]$
and of length $D/2$ if the interruption is in the interval $(0,2^{k-\lambda}]$.
We now give a lower bound on $\expmul[\mu']{X(\lambda)}$ based on $\expmul[\mu^+]{X(\lambda)}$. 
For that, we upper bound, by $D$, the loss in the potential profit due to probability mass crossing
the end of a contract of $X(\lambda)$ because of $\psi^-$. Indeed, any contract 
completed by $X(\lambda)$ before time $2D$ has length at most $D$.

\begin{align} \nonumber
\expmul[\mu']{X(\lambda)} &\geq \expmul[\mu^+]{X(\lambda)}\\ \nonumber
-& D\cdot\int_{y=0}^{2^{k-1-\lambda}} \int_{x=D}^{2^{k-\lambda}}  \psi^-(x,y)\mathrm dx \mathrm dy \\
-& D\cdot\int_{y=0}^{2^{k-\lambda}} \int_{x=2^{k-\lambda}}^{2D} \psi^-(x,y)\mathrm dx \mathrm dy.
\label{eq:emdexmuprime}
\end{align} 

Fixing the value of $\emd(\mu^+,\mu') $, the RHS of~\eqref{eq:emdexmuprime} is minimized if $\psi^-(x,2^{k-1-\lambda}) = \mu(x)$ for
$x\in[D,D+\delta_1]$ and $\psi^-(x,2^{k-\lambda}) =
\mu(x)$ for $x\in[2^{k-\lambda},2^{k-\lambda}+\delta_2]$, for some
$\delta_1,~\delta_2$, and $\psi^-(x,y)=0$ for any other $x\neq y$. 
We now upper bound $\delta_1$ (resp. $\delta_2$), using $\eta\geq\emd(\mu^+,\mu')$. 
Since $\mu(x)\geq\mu(2D)= 1/(2D)$, we have:
\begin{align}
\nonumber
\eta &\geq \int_{D}^{D+\delta_1} |x-2^{k-1-\lambda}| \cdot\psi^-(x,2^{k-1-\lambda}) ~ \mathrm dx\\\nonumber
&\geq \int_{D+\delta_1/2}^{D+\delta_1} \frac{\delta_1}2 \cdot \mu(2D) ~ \mathrm dx  \\\nonumber
&\geq \frac{\delta_1}2 \cdot\frac{\delta_1}2 \cdot \frac 1{2D}\\
\delta_1 &\leq \sqrt{8\eta D}.
\label{eq:emddelta}
\end{align}
We have the same result for $\delta_2$, so both $\delta_1$ and $\delta_2$ are
smaller than $\sqrt{8\eta\cdot D}$.

We now upper bound the decrease in $\expmul[\mu']{X(\lambda)}$ that follows from 
Equation~\eqref{eq:emddelta}. We can bound the integral:
\begin{align}\nonumber 
\int_{2^{k-1-\lambda}}^{2^{k-1-\lambda}+\delta_1}\mu(x) &\leq \delta_1 \cdot \mu(D) \\ \nonumber
&\leq  \frac 2{D} \cdot \sqrt{8\eta\cdot D} \\
&\leq 4\sqrt{\frac{2\eta}{D}}, \nonumber 
\end{align}
and equivalently for $\delta_2$. Therefore, using Equation~\eqref{eq:emdexmuprime},  we obtain Equation~\eqref{eq:emdexmuprime2MAIN} from the main paper:
\begin{align}\nonumber
\expmul[\mu']{X(\lambda)} &\geq 
\expmul[\mu^+]{X(\lambda)} - 2D\cdot 4\sqrt{\frac{2\eta}{D}}\\ 
&\geq 
\expmul[\mu^+]{X(\lambda)}-8\sqrt{{2\eta D}}.
\tag{\ref{eq:emdexmuprime2MAIN}}
\end{align}

To conclude the proof, note that $\expmu=2D\ln 2$
and $\expmul{X(\lambda)}=D/2$, from the proof of Theorem~\ref{thm:lower.dist.cont}. 
Combing Equations~\eqref{eq:emdemuxpMAIN}, \eqref{eq:emdemupMAIN}, \eqref{eq:emdemuprimeMAIN} and \eqref{eq:emdexmuprime2MAIN}, we obtain for  $\eta<D/512$ (to ensure there is no division by zero):

\begin{align*}
\frac{\expmu[\mu']}{\expmul[\mu']{X(\lambda)}} 
&\leq \frac{\expmu[\mu^+]}{\expmul[\mu^+]{X(\lambda)}-8\sqrt{2\eta D}}\\
&\leq \frac{\expmu[\mu]+\eta}{\expmul{X(\lambda)}-8\sqrt{2\eta D}}\\
&\leq \frac{2D\ln 2 + \eta}{\frac D2-8\sqrt{2\eta D}}\\
&\leq \frac{4\ln 2 + 2\eta/D}{1-16\sqrt{2\eta /D}}.
\end{align*} 

Finally, for small values of $\eta/D$, we have:
$$\frac{\expmu[\mu']}{\expmul[\mu']{X(\lambda)}} \leq 4\ln 2 + O(\sqrt{\eta/D}),$$
which concludes the proof.
\end{proof}

%% file: main.bbl
\begin{thebibliography}{}

\bibitem[\protect\citeauthoryear{Almanza \bgroup \em et al.\egroup
  }{2021}]{almanza2021online}
Matteo Almanza, Flavio Chierichetti, Silvio Lattanzi, Alessandro Panconesi, and
  Giuseppe Re.
\newblock Online facility location with multiple advice.
\newblock {\em Advances in Neural Information Processing Systems},
  34:4661--4673, 2021.

\bibitem[\protect\citeauthoryear{Alpern and Gal}{2003}]{searchgames}
Steve Alpern and Shmuel Gal.
\newblock {\em The theory of search games and rendezvous}.
\newblock Kluwer Academic Publishers, 2003.

\bibitem[\protect\citeauthoryear{Anand \bgroup \em et al.\egroup
  }{2021}]{anand2021regression}
Keerti Anand, Rong Ge, Amit Kumar, and Debmalya Panigrahi.
\newblock A regression approach to learning-augmented online algorithms.
\newblock {\em Advances in Neural Information Processing Systems},
  34:30504--30517, 2021.

\bibitem[\protect\citeauthoryear{Anand \bgroup \em et al.\egroup
  }{2022}]{anand2022online}
Keerti Anand, Rong Ge, Amit Kumar, and Debmalya Panigrahi.
\newblock Online algorithms with multiple predictions.
\newblock In {\em International Conference on Machine Learning}, pages
  582--598. PMLR, 2022.

\bibitem[\protect\citeauthoryear{Angelopoulos and
  Jin}{2019}]{DBLP:conf/ijcai/0001J19}
Spyros Angelopoulos and Shendan Jin.
\newblock Earliest-completion scheduling of contract algorithms with end
  guarantees.
\newblock In {\em Proceedings of the 28th International Joint Conference on
  Artificial Intelligence, {(IJCAI)}}, pages 5493--5499, 2019.

\bibitem[\protect\citeauthoryear{Angelopoulos and
  Kamali}{2023}]{DBLP:journals/jair/AngelopoulosK23}
Spyros Angelopoulos and Shahin Kamali.
\newblock Contract scheduling with predictions.
\newblock {\em J. Artif. Intell. Res.}, 77:395--426, 2023.

\bibitem[\protect\citeauthoryear{Angelopoulos and
  L\'{o}pez-Ortiz}{2009}]{ALO:multiproblem}
Spyros Angelopoulos and Alejandro L\'{o}pez-Ortiz.
\newblock Interruptible algorithms for multi-problem solving.
\newblock In {\em Proceedings of the 21st International Joint Conference on
  Artificial Intelligence (IJCAI)}, pages 380--386, 2009.

\bibitem[\protect\citeauthoryear{Angelopoulos and
  Panagiotou}{2023}]{DBLP:journals/jcss/AngelopoulosP23}
Spyros Angelopoulos and Konstantinos Panagiotou.
\newblock Weighted online search.
\newblock {\em J. Comput. Syst. Sci.}, 138:103457, 2023.

\bibitem[\protect\citeauthoryear{Angelopoulos \bgroup \em et al.\egroup
  }{2008}]{soft-contracts}
Spyros Angelopoulos, Alejandro L\'opez-Ortiz, and Angele Hamel.
\newblock Optimal scheduling of contract algorithms with soft deadlines.
\newblock In {\em Proceedings of the 23rd AAAI Conference on Artificial
  Intelligence (AAAI)}, pages 868--873, 2008.

\bibitem[\protect\citeauthoryear{Angelopoulos \bgroup \em et al.\egroup
  }{2020}]{DBLP:conf/innovations/0001DJKR20}
Spyros Angelopoulos, Christoph D{\"{u}}rr, Shendan Jin, Shahin Kamali, and
  Marc~P. Renault.
\newblock Online computation with untrusted advice.
\newblock In {\em Proceedings of the 11th International Conference on
  Innovations in Theoretical Computer Science (ITCS)}, pages 52:1--52:15, 2020.

\bibitem[\protect\citeauthoryear{Angelopoulos}{2015}]{spyros:ijcai15}
Spyros Angelopoulos.
\newblock Further connections between contract-scheduling and ray-searching
  problems.
\newblock In {\em Proceedings of the 24th International Joint Conference on
  Artificial Intelligence (IJCAI)}, pages 1516--1522, 2015.

\bibitem[\protect\citeauthoryear{Antoniadis \bgroup \em et al.\egroup
  }{2023}]{antoniadis2023mixing}
Antonios Antoniadis, Christian Coester, Marek Elias, Adam Polak, and Bertrand
  Simon.
\newblock Mixing predictions for online metric algorithms.
\newblock In {\em International Conference on Machine Learning}, pages
  969--983. PMLR, 2023.

\bibitem[\protect\citeauthoryear{Aspnes \bgroup \em et al.\egroup
  }{1997}]{DBLP:journals/jacm/AspnesAFPW97}
James Aspnes, Yossi Azar, Amos Fiat, Serge~A. Plotkin, and Orli Waarts.
\newblock On-line routing of virtual circuits with applications to load
  balancing and machine scheduling.
\newblock {\em J. {ACM}}, 44(3):486--504, 1997.

\bibitem[\protect\citeauthoryear{Bernstein \bgroup \em et al.\egroup
  }{2002}]{BPZF.2002.scheduling}
Daniel~S. Bernstein, T.~J. Perkins, Shlomo Zilberstein, and Lev Finkelstein.
\newblock Scheduling contract algorithms on multiple processors.
\newblock In {\em Proceedings of the 18th AAAI Conference on Artificial
  Intelligence (AAAI)}, pages 702--706, 2002.

\bibitem[\protect\citeauthoryear{Bernstein \bgroup \em et al.\egroup
  }{2003}]{steins}
Daniel~S. Bernstein, Lev Finkelstein, and Shlomo Zilberstein.
\newblock Contract algorithms and robots on rays: Unifying two scheduling
  problems.
\newblock In {\em Proceedings of the 18th International Joint Conference on
  Artificial Intelligence (IJCAI)}, pages 1211--1217, 2003.

\bibitem[\protect\citeauthoryear{Bienkowski \bgroup \em et al.\egroup
  }{2021}]{DBLP:journals/tcs/BienkowskiBBCDF21}
Marcin Bienkowski, Martin B{\"{o}}hm, Jaroslaw Byrka, Marek Chrobak, Christoph
  D{\"{u}}rr, Luk\'a\v{s} Folwarczn\'y, Lukasz Jez, Jir{\'{\i}} Sgall,
  Kim~Thang Nguyen, and Pavel Vesel{\'{y}}.
\newblock New results on multi-level aggregation.
\newblock {\em Theor. Comput. Sci.}, 861:133--143, 2021.

\bibitem[\protect\citeauthoryear{Boddy and
  Dean}{1994}]{deliberation:boddy.dean}
Mark Boddy and Thomas~L. Dean.
\newblock Deliberation scheduling for problem solving in time-constrained
  environments.
\newblock {\em Artif. Intell.}, 67(2):245--285, 1994.

\bibitem[\protect\citeauthoryear{Chrobak and Kenyon{-}Mathieu}{2006}]{ChrKen06}
Marek Chrobak and Claire Kenyon{-}Mathieu.
\newblock {SIGACT} news online algorithms column 10: Competitiveness via
  doubling.
\newblock {\em {SIGACT} News}, 37(4):115--126, 2006.

\bibitem[\protect\citeauthoryear{Chrobak \bgroup \em et al.\egroup
  }{2008}]{DBLP:journals/algorithmica/ChrobakKNY08}
Marek Chrobak, Claire Kenyon, John Noga, and Neal~E. Young.
\newblock Incremental medians via online bidding.
\newblock {\em Algorithmica}, 50(4):455--478, 2008.

\bibitem[\protect\citeauthoryear{Demaine \bgroup \em et al.\egroup
  }{2006}]{demaine:turn}
E.D. Demaine, S.P. Fekete, and S.~Gal.
\newblock Online searching with turn cost.
\newblock {\em Theoretical Computer Science}, 361:342--355, 2006.

\bibitem[\protect\citeauthoryear{Diakonikolas \bgroup \em et al.\egroup
  }{2021}]{diakonikolas2021learning}
Ilias Diakonikolas, Vasilis Kontonis, Christos Tzamos, Ali Vakilian, and Nikos
  Zarifis.
\newblock Learning online algorithms with distributional advice.
\newblock In {\em International Conference on Machine Learning}, pages
  2687--2696. PMLR, 2021.

\bibitem[\protect\citeauthoryear{Dinitz \bgroup \em et al.\egroup
  }{2022}]{dinitz2022algorithms}
Michael Dinitz, Sungjin Im, Thomas Lavastida, Benjamin Moseley, and Sergei
  Vassilvitskii.
\newblock Algorithms with prediction portfolios.
\newblock {\em Advances in neural information processing systems},
  35:20273--20286, 2022.

\bibitem[\protect\citeauthoryear{Ebenlendr and
  Sgall}{2009}]{DBLP:journals/scheduling/EbenlendrS09}
Tom{\'{a}}s Ebenlendr and Jir{\'{\i}} Sgall.
\newblock Optimal and online preemptive scheduling on uniformly related
  machines.
\newblock {\em J. Sched.}, 12(5):517--527, 2009.

\bibitem[\protect\citeauthoryear{Eberle \bgroup \em et al.\egroup
  }{2022}]{DBLP:conf/aaai/EberleLMNS22}
Franziska Eberle, Alexander Lindermayr, Nicole Megow, Lukas N{\"{o}}lke, and
  Jens Schl{\"{o}}ter.
\newblock Robustification of online graph exploration methods.
\newblock In {\em Proceedings of the 36th {AAAI} Conference on Artificial
  Intelligence, {AAAI} 2022}, pages 9732--9740. {AAAI} Press, 2022.

\bibitem[\protect\citeauthoryear{Gollapudi and
  Panigrahi}{2019}]{gollapudi2019online}
Sreenivas Gollapudi and Debmalya Panigrahi.
\newblock Online algorithms for rent-or-buy with expert advice.
\newblock In {\em Proceedings of the 36th International Conference on Machine
  Learning {(ICML)}}, pages 2319--2327, 2019.

\bibitem[\protect\citeauthoryear{Horvitz}{1988}]{Horvitz:reasoning}
Eric Horvitz.
\newblock Reasoning about beliefs and actions under computational resource
  constraints.
\newblock {\em Int. J. Approx. Reasoning}, 2(3):337--338, 1988.

\bibitem[\protect\citeauthoryear{Im \bgroup \em et al.\egroup
  }{2023a}]{DBLP:conf/infocom/ImMXZ23}
Sungjin Im, Benjamin Moseley, Chenyang Xu, and Ruilong Zhang.
\newblock Online dynamic acknowledgement with learned predictions.
\newblock In {\em {INFOCOM}}, pages 1--10. {IEEE}, 2023.

\bibitem[\protect\citeauthoryear{Im \bgroup \em et al.\egroup
  }{2023b}]{DBLP:conf/pkdd/ImMXZ23}
Sungjin Im, Benjamin Moseley, Chenyang Xu, and Ruilong Zhang.
\newblock Online state exploration: Competitive worst case and
  learning-augmented algorithms.
\newblock In {\em {ECML/PKDD} {(4)}}, volume 14172 of {\em Lecture Notes in
  Computer Science}, pages 333--348. Springer, 2023.

\bibitem[\protect\citeauthoryear{Jaillet and Stafford}{1993}]{jaillet:online}
P.~Jaillet and M.~Stafford.
\newblock Online searching.
\newblock {\em Operations Research}, 49:234--244, 1993.

\bibitem[\protect\citeauthoryear{Lindermayr and Megow}{2023}]{predictionslist}
Alexander Lindermayr and Nicole Megow.
\newblock Repository of works on algorithms with predictions.
\newblock \url{https://algorithms-with-predictions.github.io}, 2023.
\newblock Accessed: 2023-12-01.

\bibitem[\protect\citeauthoryear{L\'opez-Ortiz \bgroup \em et al.\egroup
  }{2014}]{aaai06:contracts}
Alejandro L\'opez-Ortiz, Spyros Angelopoulos, and Angele Hamel.
\newblock Optimal scheduling of contract algorithms for anytime
  problem-solving.
\newblock {\em J. Artif. Intell. Res.}, 51:533--554, 2014.

\bibitem[\protect\citeauthoryear{Lykouris and
  Vassilvitskii}{2021}]{DBLP:journals/jacm/LykourisV21}
Thodoris Lykouris and Sergei Vassilvitskii.
\newblock Competitive caching with machine learned advice.
\newblock {\em J. {ACM}}, 68(4):24:1--24:25, 2021.

\bibitem[\protect\citeauthoryear{Mitzenmacher and
  Vassilvitskii}{2020}]{DBLP:books/cu/20/MitzenmacherV20}
Michael Mitzenmacher and Sergei Vassilvitskii.
\newblock Algorithms with predictions.
\newblock In {\em Beyond the Worst-Case Analysis of Algorithms}, pages
  646--662. Cambridge University Press, 2020.

\bibitem[\protect\citeauthoryear{Purohit \bgroup \em et al.\egroup
  }{2018}]{NIPS2018_8174}
Manish Purohit, Zoya Svitkina, and Ravi Kumar.
\newblock Improving online algorithms via {ML} predictions.
\newblock In {\em Proceedings of the 31st Annual Conference on Neural
  Information Processing Systems {(NIPS)}}, pages 9661--9670, 2018.

\bibitem[\protect\citeauthoryear{Rubner \bgroup \em et al.\egroup
  }{1998}]{rubner1998metric}
Yossi Rubner, Carlo Tomasi, and Leonidas~J Guibas.
\newblock A metric for distributions with applications to image databases.
\newblock In {\em Sixth international conference on computer vision (IEEE Cat.
  No. 98CH36271)}, pages 59--66. IEEE, 1998.

\bibitem[\protect\citeauthoryear{Russell and
  Zilberstein}{1991}]{RZ.1991.composing}
Stuart~J. Russell and Shlomo Zilberstein.
\newblock Composing real-time systems.
\newblock In {\em Proceedings of the 12th International Joint Conference on
  Artificial Intelligence (IJCAI)}, pages 212--217, 1991.

\bibitem[\protect\citeauthoryear{Sun \bgroup \em et al.\egroup
  }{2021}]{sun2021pareto}
Bo~Sun, Russell Lee, Mohammad Hajiesmaili, Adam Wierman, and Danny Tsang.
\newblock Pareto-optimal learning-augmented algorithms for online conversion
  problems.
\newblock In {\em Proceedings of the 34th Annual Conference on Neural
  Information Processing Systems ({NeurIPS})}, pages 10339--10350, 2021.

\bibitem[\protect\citeauthoryear{Sung and Tokekar}{2019}]{sung2019competitive}
Yoonchang Sung and Pratap Tokekar.
\newblock A competitive algorithm for online multi-robot exploration of a
  translating plume.
\newblock In {\em 2019 International Conference on Robotics and Automation
  (ICRA)}, pages 3391--3397. IEEE, 2019.

\bibitem[\protect\citeauthoryear{Wang \bgroup \em et al.\egroup
  }{2020}]{wang2020online}
Shufan Wang, Jian Li, and Shiqiang Wang.
\newblock Online algorithms for multi-shop ski rental with machine learned
  advice.
\newblock {\em Advances in Neural Information Processing Systems},
  33:8150--8160, 2020.

\bibitem[\protect\citeauthoryear{Zilberstein and
  Russell}{1993}]{DBLP:conf/ijcai/ZilbersteinR93}
Shlomo Zilberstein and Stuart Russell.
\newblock Anytime sensing planning and action: {A} practical model for robot
  control.
\newblock In {\em {IJCAI}}, pages 1402--1407. Morgan Kaufmann, 1993.

\bibitem[\protect\citeauthoryear{Zilberstein and
  Russell}{1996}]{DBLP:journals/ai/ZilbersteinR96}
Shlomo Zilberstein and Stuart~J. Russell.
\newblock Optimal composition of real-time systems.
\newblock {\em Artif. Intell.}, 82(1-2):181--213, 1996.

\bibitem[\protect\citeauthoryear{Zilberstein \bgroup \em et al.\egroup
  }{2003}]{ZilbersteinCC03}
Shlomo Zilberstein, Francois Charpillet, and Philippe Chassaing.
\newblock Optimal sequencing of contract algorithms.
\newblock {\em Ann. Math. Artif. Intell.}, 39(1-2):1--18, 2003.

\end{thebibliography}
